%% file: IntermediateLinearizability.tex
\documentclass[10pt, a4paper]{article}

\input{environment/mypackages.tex}

\input{environment/mystyle.tex}

\begin{document}

\title{Intermediate Value Linearizability: A Quantitative Correctness Criterion} 

\author{Arik Rinberg
\\ ArikRinberg@campus.technion.ac.il
\\
\\ Idit Keider
\\ idish@ee.technion.ac.il
\\
\\Technion - Israel Institute of Technology}

\date{}

\maketitle

\input{sections/abstract.tex}

\input{sections/intro.tex}

\input{sections/preliminaries.tex}

\input{sections/intermediateLinearizability.tex}

\input{sections/boundedObjects.tex}

\input{sections/countMin.tex}

\input{sections/adderObject.tex}

\input{sections/lowerBound.tex}

\input{sections/conclusion.tex}

\bibliographystyle{plain}
\bibliography{bibliography}

\newpage

\appendix

\input{sections/appendix.tex}

\end{document}

%% file: environment/mypackages.tex
\usepackage[utf8]{inputenc}
\usepackage{amsmath}
\usepackage{amssymb}
\usepackage{amsthm}
\usepackage{graphicx}
\usepackage{algorithm}
\usepackage[noend]{algpseudocode}
\usepackage{xcolor}
\usepackage{framed}
\usepackage{bm}
\usepackage{pifont}
\usepackage{multicol}
\usepackage{lipsum}
\usepackage{tikz}
\usetikzlibrary{calc}
\usepackage{wrapfig}
\usepackage{array}
\usepackage{mathtools}
\usepackage{url}
\usepackage{thm-restate}
\usepackage{varwidth}
\usepackage{dsfont}

\usepackage{esvect}
\usepackage{enumitem}
\usepackage[strict]{changepage}

\usepackage[pdftex]{hyperref}
\usepackage{caption}
\usepackage{subcaption}
\usepackage{anyfontsize}

\usepackage{afterpage}
\usepackage{hypbmsec}

%% file: environment/mystyle.tex
\newtheorem{invariant}{Invariant}
\newtheorem{definition}{Definition}

\newtheorem{theorem}{Theorem}
\newtheorem{lemma}{Lemma}
\newtheorem{example}{Example}
\newtheorem{corollary}{Corollary}

\newcommand{\remove}[1]{}

\newcommand{\CM}[4]{
    \begin{tikzpicture}[baseline=3ex]
        \draw[step=0.5cm,color=gray] (0,0) grid (1,1);
        \node at (0.25,0.25) {#3};
        \node at (0.25,0.75) {#1};
        \node at (0.75,0.75) {#2};
        \node at (0.75,0.25) {#4};
    \end{tikzpicture}
}

%% file: sections/abstract.tex
\begin{abstract}
    
Big data processing systems often employ batched updates and data sketches to
estimate certain properties of large data. For example, a \emph{CountMin sketch}
approximates the frequencies at which elements occur in a data stream, and a
\emph{batched counter} counts events in batches. This paper focuses on correctness
criteria for concurrent implementations of such objects. Specifically, we consider
\emph{quantitative} objects, whose return values are from a totally ordered domain,
with a particular emphasis on \emph{$(\epsilon,\delta)$-bounded} objects that
estimate a numerical quantity with an error of at most $\epsilon$ with probability
at least $1 - \delta$.

The de facto correctness criterion for concurrent objects is linearizability.
Intuitively, under linearizability, when a read overlaps an update, it must
return the object's value either before the update or after it. Consider,
for example, a single batched increment operation that counts three new events,
bumping a batched counter's value from $7$ to $10$. In a linearizable implementation of the counter,
a read overlapping this update must return either $7$ or $10$. We observe, however, that in typical
use cases, any \emph{intermediate} value between $7$ and $10$ would also be
acceptable. To capture this additional degree of freedom, we propose
\emph{Intermediate Value Linearizability (IVL)}, a new correctness
criterion that relaxes linearizability to allow returning
intermediate values, for instance $8$ in the example above. Roughly
speaking, IVL allows reads to return any value that is bounded
between two return values that are legal under linearizability. 
A key feature of IVL is that we can prove that concurrent IVL implementations of
$(\epsilon,\delta)$-bounded objects are themselves $(\epsilon,\delta)$-bounded.
To illustrate the power of this result, we give a straightforward
and efficient concurrent implementation of an $(\epsilon, \delta)$-bounded
CountMin sketch, which is IVL (albeit not linearizable). 

Finally, we show that IVL allows for inherently cheaper implementations
than linearizable ones. In particular, we show a lower bound of $\Omega(n)$
on the step complexity of the update operation of any wait-free linearizable
batched counter, and propose a wait-free IVL implementation of the same
object with an $O(1)$ step complexity for update.
\end{abstract}

%% file: sections/intro.tex
\section{Introduction}
\label{sec:intro}

\subsection{Motivation}
\label{ssec:motivation}

Big data processing systems often perform analytics on incoming data streams,
and must do so at a high rate due to the speed of incoming data.
Data sketching algorithms, or \emph{sketches} for short~\cite{cormode2012synopses}, are
an indispensable tool for such high-speed computations. Sketches typically estimate some function
of a large stream, for example, the frequency of certain items~\cite{cormode2005improved}, how many unique items
have appeared~\cite{datar2002comparing, flajolet1983probabilistic, gibbons2001estimating},
or the top-$k$ most common items~\cite{metwally2005efficient}.
They are supported by many data analytics platforms such as PowerDrill~\cite{heule2013hyperloglog},
Druid~\cite{druid}, Hillview~\cite{hillview}, and Presto~\cite{presto} as well as standalone toolkits~\cite{apache-datasketches}.

Sketches are quantitative objects that support {\sc update}
and {\sc query} operations, where the return value of a {\sc query}
is from a totally ordered set. They are essentially succinct (sublinear)
summaries of a data stream. For example, a sketch might estimate the number of packets originating
from any IP address, without storing a record for every packet.
Typical sketches are \emph{probably approximately correct (PAC)}, estimating some aggregate
quantity with an error of at most $\epsilon$ with probability at least $1-\delta$ for some
parameters $\epsilon$ and $\delta$.


The ever increasing rates of incoming data create a strong demand for parallel
stream processing~\cite{cormode2011algorithms,heule2013hyperloglog}.
In order to allow queries to return fresh results in real-time without
hampering data ingestion, it is paramount to support queries concurrently with updates~\cite{rinberg2019fast,stylianopoulos2020delegation}.
But parallelizing sketches raises some important questions, for instance: \textit{What are the semantics of overlapping operations in a concurrent sketch?},
\textit{How can we prove error guarantees for such a sketch?}, and, in particular,
\textit{Can we reuse the myriad of clever analyses of existing sketches' error bounds in parallel settings without opening the black box?}
In this paper we address these questions.


\subsection{Our contributions}
\label{ssec:contribution}

The most common correctness condition for concurrent objects is linearizability. Roughly speaking,
it requires each parallel execution to have a \emph{linearization}, which is a sequential
execution of the object that ``looks like'' the parallel one. (See Section~\ref{sec:preliminaries} for a formal definition.)
But sometimes linearizability is too restrictive, leading to a high implementation cost.

In Section~\ref{sec:ivl}, we propose \emph{Intermediate Value Linearizability (IVL)},
a new correctness criterion for quantitative objects.
Intuitively, the return value of an operation
of an IVL object is bounded between two legal values that can be returned in linearizations.
The motivation for allowing this is that if the system designer is happy with
either of the legal values, then the intermediate value should also be fine.
For example, consider a system where processes
count events, and a monitoring process detects when the number of events passes a threshold.
The monitor constantly reads a shared counter, which other process increment in batches.
If an operation increments the counter from $4$ to $7$
batching three events, IVL allows a concurrent read by the monitoring
process to return $6$, although there is no linearization
in which the counter holds $6$. We formally define IVL and prove that this property is
\emph{local}, meaning that a history composed of IVL objects is itself IVL.
This allows reasoning about single objects rather than about the system as a whole. We formulate
IVL first for sequential objects, and then extend it to capture randomized ones.

Sketching algorithms have sequential error analyses which we wish to leverage
for the concurrent case. In Section~\ref{sec:bounded-objects} we formally define $(\epsilon, \delta)$-bounded
objects, including concurrent ones. We then prove a key theorem about IVL, stating that an IVL
implementation of a sequential $(\epsilon, \delta)$-bounded object
is itself $(\epsilon, \delta)$-bounded. The importance of this theorem is that it provides a generic way to leverage
the vast literature on sequential $(\epsilon, \delta)$-bounded
sketches~\cite{morris1978counting, flajolet1985approximate, cichon2011approximate, liu2016one, cormode2005improved, agarwal2013mergeable}
in concurrent implementations.

As an example, in Section~\ref{sec:countMin}, we present a concurrent CountMin sketch~\cite{cormode2005improved},
which estimates the frequencies of items in a data stream. We prove that a straightforward
parallelization of this sketch is IVL. By the aforementioned theorem, we deduce that the concurrent sketch adheres
to the error guarantees of the original sequential one, without having to ``open'' the analysis. We note
that this parallelization is \emph{not} linearizable.

Finally, we show that IVL is sometimes inherently cheaper than linearizability. We illustrate
this in Section~\ref{sec:adder} via the example of a \emph{batched counter}. We present
a wait-free IVL implementation of this object from
single-writer-multi-reader (SWMR) registers with $O(1)$ step complexity for 
{\sc update} operations. We then prove a lower bound of $\Omega(n)$ step
complexity for the {\sc update} operation of any wait-free
linearizable implementation, using only SWMR registers.
This exemplifies that there is an inherent and unavoidable
cost when implementing linearizable algorithms, which can be circumvented
by implementing IVL algorithms instead.

%% file: sections/preliminaries.tex
\section{Preliminaries}
\label{sec:preliminaries}

Section~\ref{ssec:det-objects} discusses deterministic shared memory objects and defines linearizability.
In Section~\ref{ssec:rand-objects} we discuss randomized algorithms and their correctness criteria.

\subsection{Deterministic objects}
\label{ssec:det-objects}

We consider a standard shared memory model~\cite{herlihy1990linearizability}, where a set of
\emph{processes} access atomic shared memory variables. Accessing these
shared variables is instantaneous.
Processes take \emph{steps} according to an \emph{algorithm}, which is a deterministic
state machine, where a step can access a shared memory variable, do local computations, and possibly return
some value. An \emph{execution} of an algorithm is an alternating
sequence of steps and states. We focus on algorithms that
implement \emph{objects}, which support
\emph{operations}, such as {\sc read} and {\sc write}. Operations begin
with an \emph{invocation} step and end with a \emph{response} step.
A \emph{schedule}, denoted $\sigma$, is the order in
which processes take steps, and the operations they invoke
in invoke steps with their parameters. Because we consider deterministic algorithms,
$\sigma$ uniquely defines an execution of a given algorithm.

A \emph{history} is the sequence of invoke and response steps in an execution. Given an algorithm $A$ and
a schedule $\sigma$, $H(A, \sigma)$ is the history of the execution of $A$ with
schedule $\sigma$. A \emph{sequential} history is an alternating sequence of
invocations and their responses, beginning with an invoke step.
We denote the return value of operation $op$ with parameter $arg$ in history $H$ by $\text{ret}(op,H)$.
We refer to the invocation step of operation $op$ with parameter $arg$
by process $p$ as $inv_p(op(arg))$ and to its response
step by $rsp_p(op)\rightarrow ret$, where $ret = \text{ret}(op,H)$. A history
defines a partial order on operations: Operation $op_1$ \emph{precedes} $op_2$
in history $H$, denoted $op_1 \prec_H op_2$, if $rsp(op_1)$ precedes $inv(op_2(arg))$
in $H$. Two operations are \emph{concurrent} if neither precedes the other.

A \emph{well-formed} history is one that does not contain concurrent operations by the same process, and
where every response event for operation $op$ is preceded by an invocation of the same operation.
A schedule is well-formed if it gives rise to a well-formed history, and an execution
is well-formed if it is based on a well-formed schedule.
We denote by $H|_x$ the sub-history of $H$ consisting only of invocations and responses
on object $x$. Operation $op$ is \emph{pending} in a history $H$ if $op$ is invoked
in $H$ but does not return.

Correctness of an object's implementation is defined with respect to a
sequential specification $\mathcal{H}$, which is the object's set of allowed sequential histories.
If the history spans multiple objects, $\mathcal{H}$ consists of sequential histories $H$ such that for all
objects $x$, $H|_x$ pertains to $x$'s sequential specification (denoted $\mathcal{H}_x$).
A \emph{linearization}~\cite{herlihy1990linearizability} of a concurrent history $H$ is a
sequential history $H'$ such that (1) after removing some pending
operations from $H$ and completing others, it contains the same invocations and
responses as $H'$ with the same parameters and return values, and (2) $H'$
preserves the partial order $\prec_H$.
Algorithm $A$ is a \emph{linearizable implementation}
of a sequential specification $\mathcal{H}$ if every history of a
well-formed execution of $A$ has a linearization in $\mathcal{H}$.




\subsection{Randomized algorithms}
\label{ssec:rand-objects}


In randomized algorithms, processes have access to coin flips from some domain $\Omega$.
Every execution is associated with a coin flip vector $\vv{c}=(c_1, c_2, \dots)$,
where $c_i \in \Omega$ is the $i^\text{th}$ coin flip in the execution.
A \emph{randomized algorithm} $A$ is a probability distribution over deterministic
algorithms $\{A({\vv{c})}\}_{{\vv{c}} \in \Omega^{\infty}}$\footnote{We do not consider non-deterministic objects in this paper.},
arising when $A$ is instantiated with different coin flip vectors.
We denote by $H(A, {\vv{c}}, \sigma)$ the history of the execution of randomized algorithm
$A$ observing coin flip vector $\vv{c}$ in schedule $\sigma$.



Golab et al. show that randomized algorithms that use concurrent objects require a stronger
correctness criterion than linearizability, and propose \emph{strong linearizability}~\cite{golab2011linearizable}.
Roughly speaking, strong linearizability stipulates that the
mapping of histories to linearizations must be prefix-preserving,
so that future coin flips cannot impact the
linearization order of earlier events.
In contrast to us, they consider deterministic objects used by randomized algorithms. In this paper,
we focus on randomized object implementations.

%% file: sections/intermediateLinearizability.tex
\section{Intermediate value linearizability}
\label{sec:ivl}

Section~\ref{ssec:definitions} introduces definitions that we utilize to define IVL.
Section~\ref{ssec:ivl} defines IVL for deterministic algorithms and proves that it is a local property.
Section~\ref{ssec:sivl} extends IVL for randomized algorithms, and Section~\ref{ssec:comparisons}
compares IVL to other correctness criteria.

\subsection{Definitions}
\label{ssec:definitions}

Throughout this paper we consider the strongest progress guarantee, bounded wait-freedom.
An operation $op$ is \emph{bounded wait-free} if whenever any process $p$
invokes $op$, $op$ returns a response in a bounded number of
$p$'s steps, regardless of steps taken by other processes.
An operation's \emph{step-complexity} is the maximum number of steps
a process takes during a single execution of this operation. We
can convert every bounded wait-free algorithm to a \emph{uniform step complexity}
one, in which each operation takes the exact same
number of steps in every execution.
This can be achieved by padding shorter
execution paths with empty steps before returning.
Note that in a randomized algorithm with uniform step complexity, coin flips have
no impact on $op$'s execution times. For the remainder of this paper, we consider
algorithms with uniform step complexity.

Our definitions use the notion of skeleton histories:
A \emph{skeleton history} is a sequence of invocation and response events, where the return values
of the responses are undefined, denoted $?$. For a history $H$, we define the operator $H^?$
as altering all response values in $H$ to $?$, resulting in a skeleton
history.

In this paper we formulate correctness criteria for a class of objects we call
\emph{quantitative}. These are objects that support two operations: (1) {\sc update},
which captures all mutating operations and does not return a value; and (2) {\sc query},
which returns a value from a totally ordered domain. In a \emph{deterministic quantitative object}
the return values of {\sc query} operations are uniquely defined. Namely, the object's sequential specification $\mathcal{H}$
contains exactly one history for every sequential history skeleton $H$; we denote this
history by $\tau_\mathcal{H}(H)$. Thus, $\tau_\mathcal{H}(H^?) = H$ for every history $H\in \mathcal{H}$.
Furthermore, for every sequential skeleton history $H$, by definition, $\tau_\mathcal{H}(H) \in \mathcal{H}$.

\begin{example}

Consider an execution in which a batched counter initialized to $0$
is incremented by $3$ by process $p$ concurrently with a query by
process $q$, which returns $0$. Its history is:
\[ H = inv_p(inc(3)), inv_q(query), rsp_p(inc), rsp_q(query \rightarrow 0). \]
The skeleton history $H^?$ is:
\[ H^? = inv_p(inc(3)), inv_q(query), rsp_p(inc), rsp_q(query \rightarrow ?). \]
A possible linearization of $H^?$ is:
\[ H'=inv_p(inc(3)), rsp_p(inc), inv_q(query), rsp_q(query \rightarrow ?). \]
Given the sequential specification $\mathcal{H}$ of a batched counter,
we get:
\[ \tau_\mathcal{H}(H')=inv_p(inc(3)), rsp_p(inc), inv_q(query), rsp_q(query \rightarrow 3). \]
In a different linearization, the query may return $0$ instead.

\end{example}

\subsection{Intermediate value linearizability}
\label{ssec:ivl}


We now define intermediate value linearizability for quantitative objects.
\begin{definition}[Intermediate value linearizability]
  A history $H$ of an object is IVL with respect to sequential specification $\mathcal{H}$ if there
  exist two linearizations $H_1, H_2$ of $H^?$ such that for every {\sc query} $Q$
  that returns in $H$,
  \[\text{ret}(Q, \tau_\mathcal{H}(H_1)) \leq \text{ret}(Q, H) \leq \text{ret}(Q, \tau_\mathcal{H}(H_2)). \]

  Algorithm $A$ is an \emph{IVL implementation} of a sequential specification $\mathcal{H}$ if every
  history of a well-formed execution of $A$ is IVL with respect to $\mathcal{H}$.
\end{definition}

Note that a linearizable object is trivially IVL, as the skeleton history of the linearization of $H$ plays the roles
of both $H_1$ and $H_2$. The following theorem, proven in Appendix~\ref{sec:locality-proof}, shows that this property is local (as defined in~\cite{herlihy1990linearizability}):
\begin{restatable}{thm}{local}
  \label{thm:ivl-local}
  A history $H$ of a well-formed execution of algorithm $A$ over a set of objects $\mathcal{X}$
  is IVL if and only if for each object $x \in \mathcal{X}$, $H|_x$ is IVL.
\end{restatable}

Locality allows system designers to reason about their system in a modular fashion. Each object can be built separately,
and the system as a whole still satisfies the property.

\subsection{Extending IVL for randomized algorithms}
\label{ssec:sivl}



In a randomized algorithm $A$ with uniform step complexity, every invocation
of a given operation returns after the same number of steps,
regardless of the coin flip vector $\vv{c}$. This, in turn, implies that
for a given $\sigma$, for any $\vv{c}, \vv{c}' \in \Omega^{\infty}$, the
arising histories $H(A, \vv{c}, \sigma)$ and $H(A, \vv{c}', \sigma)$ differ only in the operations'
return values but not in the order of invocations and responses, as the latter is
determined by $\sigma$, so their skeletons are equal. For randomized algorithm $A$ and schedule $\sigma$,
we denote this arising skeleton history by $H^?(A, \sigma)$.

We are faced with a dilemma when defining the specification
of a randomized algorithm $A$, as the algorithm itself is a distribution over a
set of algorithms $\{A(\vv{c})\}_{\vv{c}\in \Omega^\infty}$. Without knowing
the observed coin flip vector $\vv{c}$, the execution behaves unpredictably. We therefore
define a deterministic sequential specification $\mathcal{H}(\vv{c})$ for each coin flip vector
$\vv{c} \in \Omega^\infty$, so the sequential specification is a
probability distribution on a set of sequential histories $\{\mathcal{H}(\vv{c})\}_{\vv{c}\in \Omega^\infty}$.

A correctness criterion for randomized objects needs to capture the property that the distribution of
a randomized algorithm's outcomes matches the distribution of behaviors allowed by the specification.
Consider, e.g., some sequential skeleton history $H$ of an object defined by $\{\mathcal{H}(\vv{c})\}_{\vv{c}\in \Omega^\infty}$. Let Q be a query
that returns in $H$, and assume that $Q$ has some probability $p$ to return a value $v$ in $\tau_{\mathcal{H}(\vv{c})}(H)$ for
a randomly sampled $\vv{c}$. Intuitively, we would expect that if a randomized algorithm $A$ ``implements'' the specification
$\{\mathcal{H}(\vv{c})\}_{\vv{c}\in \Omega^\infty}$, then $Q$ has a similar probability to return $v$ in sequential executions of $A$ with the same history,
and to some extent also in concurrent executions of $A$ of which $H$ is a linearization. In other words, we
would like the distribution of outcomes of $A$ to match the distribution of outcomes in $\{\mathcal{H}(\vv{c})\}_{\vv{c}\in \Omega^\infty}$.

We observe that in order to achieve this, it does not suffice to require that each history have an
arbitrary linearization as we did for deterministic objects, because this might not preserve the
desired distribution. Instead, for randomized objects we require a common linearization for each
skeleton history that will hold true under all possible coin flip vectors. We therefore formally
define IVL for randomized objects as follows:

\begin{definition}[IVL for randomized algorithms]

  Consider a skeleton history $H=H^?(A, \sigma)$ of some
  randomized algorithm $A$ with schedule $\sigma$.
  $H$ is \emph{IVL} with respect to $\{\mathcal{H}(\vv{c})\}_{\vv{c} \in \Omega^\infty}$ if there exist
  linearizations $H_1, H_2$ of $H$ such that for every coin flip vector $\vv{c}$ and query $Q$
  that returns in $H$,
  \[\text{ret}(Q, \tau_{\mathcal{H}(\vv{c})}(H_1)) \leq \text{ret}(Q, H(A, \vv{c}, \sigma)) \leq \text{ret}(Q, \tau_{\mathcal{H}(\vv{c})}(H_2)). \]

  Algorithm $A$ is an \emph{IVL implementation} of a sequential specification
  distribution $\{\mathcal{H}(\vv{c})\}_{\vv{c} \in \Omega^\infty}$ if every skeleton
  history of a well-formed execution of $A$ is IVL with
  respect to $\{\mathcal{H}(\vv{c})\}_{\vv{c} \in \Omega^\infty}$.
  \label{def:sivl}
\end{definition}

Note that since we require a common linearization under all coin flip vectors, we do not need to
strengthen IVL for randomized settings in the manner that strong linearizability
strengthens linearizability. This is because the linearizations we consider are a fortiori independent of future coin flips.

\subsection{Relationship to other relaxations}
\label{ssec:comparisons}

In spirit, IVL resembles the \emph{regularity} correctness condition for
single-writer registers~\cite{lamport1986interprocess}, where a query must return
either a value written by a concurrent write or the last value written
by a write that completed before the query began. Stylianopoulos  
et al.~\cite{stylianopoulos2020delegation} adopt a similar condition for data sketches, which they
informally describe as follows: ``a query takes into account all completed
insert operations and possibly a subset of the overlapping ones.'' If
the object's estimated quantity (return value) is monotonically increasing
throughout every execution, then IVL essentially formalizes this condition,
while also allowing intermediate steps of a single update to be observed.
But this is not the case in general. Consider, for example, an
object supporting increment and decrement, and a query that occurs concurrently
with an increment and an ensuing decrement. If the query takes only the decrement
into account (and not the increment), it returns a value that is smaller than all legal return values
that may be returned in linearizations, which violates IVL. Our interval-based
formalization is instrumental to ensuring that a concurrent IVL implementation
preserves the probabilistic error bounds of the respective sequential sketch, which we prove in the next section.

Previous work on set-linearizability~\cite{neiger1994set} and
interval-linearizability~\cite{castaneda2018unifying} has also relaxed linearizability,
allowing a larger set of return values in the presence of overlapping operations. The return
values, however, must be specified in advance by a given state machine; operations' effects
on one another must be predefined. In contrast to these, IVL is generic, and does not require
additional object-specific definitions; it provides an intuitive quantitative bound
on possible return values of arbitrary quantitative objects.

Henzinger et al.~\cite{henzinger2013quantitative} define
the quantitative relaxation framework, which
allows executions to differ from the sequential specification up to a bounded cost function.
Alistarh et al.
expand upon this and define \emph{distributional linearizability}~\cite{alistarh2018distributionally},
which requires a distribution over the internal
states of the object for its error analysis.
Rinberg et al. consider strongly linearizable $r$-relaxed semantics for randomized objects~\cite{rinberg2019fast}.
We differ from these works in two points: First, a sequential history of an IVL object
must adhere to the sequential specification, whereas in these
relaxations even a sequential history may diverge from the specification. The second is that these relaxations
are measured with respect to a single linearization. We, instead, bound
the return value between two legal linearizations. The latter is the key to preserving the
error bounds of sequential objects, as we next show.




%% file: sections/boundedObjects.tex
\section{\texorpdfstring{$(\epsilon,\delta)$}{(epsilon,delta)}-bounded objects}
\label{sec:bounded-objects}

In this section we show that for a large class of randomized objects,
IVL concurrent implementations preserve the error bounds of the respective
sequential ones. More specifically, we focus on randomized
objects like data sketches, which estimate some quantity (or quantities)
with probabilistic guarantees.
Sketches generally support two operations:
{\sc update}($a$), which processes element $a$, and {\sc query}($arg$),
which returns the quantity estimated by the sketch as a function of the previously
processed elements. Sequential sketch algorithms typically have probabilistic error bounds. For example,
the Quantiles sketch estimates the rank of a given element in
a stream within $\pm \epsilon n$ of the true rank, with probability
at least $1-\delta$~\cite{agarwal2013mergeable}. 

We consider in this section a general class of
\emph{$(\epsilon, \delta)$-bounded objects} capturing PAC algorithms. A bounded object's
behavior is defined relative to a deterministic
sequential specification $\mathcal{I}$, which uniquely defines
the \emph{ideal} return value for every query in a sequential
execution. In an $(\epsilon, \delta)$-bounded $\mathcal{I}$ object,
each query returns the ideal return value within
an error of at most $\epsilon$ with probability at least $1-\delta$.
More specifically, it over-estimates (and similarly under-estimates) the
ideal quantity by at most $\epsilon$ with probability at least $1-\frac{\delta}{2}$. 
Formally:
\begin{definition}
    A sequential randomized algorithm $A$ implements an $(\epsilon,\delta)$-bounded $\mathcal{I}$
    object if for every query $Q$ returning in
    an execution of $A$ with any schedule $\sigma$ and a randomly sampled
    coin flip vector $\vv{c} \in \Omega^\infty$,
    \[ ret(Q,H(A,\sigma,\vv{c})) \geq ret(Q, \tau_{\mathcal{I}}(H^?(A,\sigma)) - \epsilon \text{ with probability at least } 1-\frac{\delta}{2},\]
    and 
    \[ ret(Q,H(A,\sigma,\vv{c})) \leq ret(Q, \tau_{\mathcal{I}}(H^?(A,\sigma)) + \epsilon \text{ with probability at least } 1-\frac{\delta}{2}.\]
    \label{def:seq-e,d-obj}
\end{definition}

A sequential algorithm $A$ satisfying Definition~\ref{def:seq-e,d-obj} induces a sequential specification
$\{A(\vv{c})\}_{\vv{c} \in \Omega^\infty}$ of an $(\epsilon,\delta)$-bounded $\mathcal{I}$ object.
We next discuss parallel implementations of this specification.

To this end, we must specify a correctness criterion on the object's
concurrent executions. As previously stated, the standard notion (for randomized algorithms)
is strong linearizability, stipulating that we can ``collapse'' each operation in the concurrent
schedule to a single point in time. Intuitively, this means that every query returns a value
that could have been returned by the randomized algorithm at some point during its execution 
interval. So the query returns an $(\epsilon, \delta)$ approximation of the ideal value at that particular point.
But this point is arbitrarily chosen, meaning that the query may return an $\epsilon$ approximation
of any value that the ideal object takes during the query's execution. We therefore look at the minimum and
maximum values that the ideal object may take during a query's interval, and bound the error relative to these values. 


We first define these minimum and maximum values as follows: For a history $H$,
denote by $\mathcal{L}(H^?)$ the set of linearizations of $H^?$.
For a query $Q$ that returns in $H$ and an ideal specification $\mathcal{I}$, we define: 
\[ v^{\mathcal{I}}_{min}(H,Q) \triangleq \min \{ \tau_\mathcal{I}(L) \mid L \in \mathcal{L}(H^?)\} ; \\
v^{\mathcal{I}}_{max}(H,Q) \triangleq \max \{ \tau_\mathcal{I}(L) \mid L \in \mathcal{L}(H^?)\}.
\]
Note that even if $H$ is infinite and has infinitely many linearizations,
because $Q$ returns in $H$, it appears in each linearization by
the end of its execution interval, and therefore $Q$ can return a finite number of different
values in these linearizations, and so the minimum and maximum are well-defined.
Correctness of concurrent $(\epsilon, \delta)$-bounded objects is then formally defined as follows:

\begin{definition}
    A concurrent randomized algorithm $A$ implements an $(\epsilon,\delta)$-bounded
    $\mathcal{I}$ object if for every query $Q$ returning in
    an execution of $A$ with any schedule $\sigma$ and a randomly sampled
    coin flip vector $\vv{c} \in \Omega^\infty$,
    \[ ret(Q,H(A,\sigma,\vv{c})) \geq  v^{\mathcal{I}}_{min}(H(A,\sigma,\vv{c}),Q) - \epsilon \text{ with probability at least } 1-\frac{\delta}{2},\]
    and 
    \[ ret(Q,H(A,\sigma,\vv{c})) \leq v^{\mathcal{I}}_{max}(H(A,\sigma,\vv{c}),Q) + \epsilon \text{ with probability at least } 1-\frac{\delta}{2}.\]
    \label{def:con-e,d-obj}
\end{definition} 

In some algorithms, $\epsilon$ depends on the stream size, i.e.,
the number of updates preceding a query; to avoid cumbersome notations we use
a single variable $\epsilon$, which should be set to the maximum
value that the sketch's $\epsilon$ bound takes during the query's execution interval.

The following theorem shows
that IVL implementations allow us to leverage the
``legacy'' analysis of a sequential object's error bounds.
\begin{theorem}
    Consider a sequential specification $\{A(\vv{c})\}_{\vv{c} \in \Omega^\infty}$ of an $(\epsilon,\delta)$-bounded
    $\mathcal{I}$ object (Definition~\ref{def:seq-e,d-obj}). Let $A'$ be an IVL implementation of $A$ (Definition~\ref{def:sivl}). Then $A'$ implements a concurrent
    $(\epsilon,\delta)$-bounded $\mathcal{I}$ object (Definition~\ref{def:con-e,d-obj}).

    \label{thm:SIVL-bound}
\end{theorem}
\begin{proof}
    Consider a skeleton history $H=H^?(A', \sigma)$ of $A'$ with some schedule $\sigma$, and a
    query $Q$ that returns in $H$. As $A'$ is an IVL implementation of $A$, there exist linearizations
    $H_1$ and $H_2$ of $H$, such that for every $\vv{c}\in\Omega^\infty$,
    $\text{ret}(Q, \tau_{A(\vv{c})}(H_1)) \leq \text{ret}(Q, H(A, \sigma, \vv{c})) \leq \text{ret}(Q, \tau_{A(\vv{c})}(H_1))$.
    As $\{A(\vv{c})\}_{\vv{c} \in \Omega^\infty}$ captures a sequential $(\epsilon,\delta)$-bounded $\mathcal{I}$ object,
    $\text{ret}(Q, \tau_{A(\vv{c})}(H_i)$ is bounded as follows:
    \[ \text{ret}(Q, \tau_{A(\vv{c})}(H_1)) \geq \text{ret}(Q, \tau_{\mathcal{I}}(H_1)) - \epsilon \text{ with probability at least } 1-\frac{\delta}{2}, \]
    and
    \[ \text{ret}(Q, \tau_{A(\vv{c})}(H_2)) \leq \text{ret}(Q, \tau_{\mathcal{I}}(H_2)) + \epsilon \text{ with probability at least } 1-\frac{\delta}{2}. \]
    Furthermore, by definition of $v_{min}$ and $v_{max}$:
    \[ \text{ret}(Q, \tau_{\mathcal{I}}(H_1)) \geq  v^{\mathcal{I}}_{min}(H(A', \sigma, \vv{c}),Q) ; \\ 
    \text{ret}(Q, \tau_{A(\vv{c})}(H_2)) \leq v^{\mathcal{I}}_{max}(H(A', \sigma, \vv{c}),Q).\]

    Therefore, with probability at least $1-\frac{\delta}{2}$, $\text{ret}(Q, H(A', \sigma, \vv{c})) \geq v^{\mathcal{I}}_{min}(H(A', \sigma, \vv{c}),Q) - \epsilon$ and
    with probability at least $1-\frac{\delta}{2}$, $\text{ret}(Q, H(A', \sigma, \vv{c})) \geq v^{\mathcal{I}}_{max}(H(A', \sigma, \vv{c}),Q) + \epsilon$, as needed.



\end{proof}

While easy to prove, Theorem~\ref{thm:SIVL-bound} shows that IVL is in some sense the ``right''
correctness property for $(\epsilon,\delta)$-bounded objects. It is less restrictive -- and as we
show below, sometimes cheaper to implement -- than linearizability, and yet strong enough
to preserve the salient properties of sequential executions of $(\epsilon,\delta)$-bounded objects. As
noted in Section~\ref{ssec:comparisons}, previously suggested relaxations do not inherently guarantee that
error bounds are preserved. For example, regular-like semantics, where a query ``sees''
some subset of the concurrent updates~\cite{stylianopoulos2020delegation},  satisfy IVL (and hence bound
the error) for monotonic objects albeit not for general ones. Indeed, if object values
can both increase and decrease, the results returned under such regular-like semantics can arbitrarily diverge from possible sequential ones.

The importance of Theorem~\ref{thm:SIVL-bound} is that it allows us to leverage
the vast literature on sequential $(\epsilon, \delta)$-bounded
objects~\cite{morris1978counting, flajolet1985approximate, cichon2011approximate, liu2016one, cormode2005improved, agarwal2013mergeable}
in concurrent implementations. As an example, in the next section we give an example of an IVL
parallelization of a popular data sketch. By Theorem~\ref{thm:SIVL-bound},
it preserves the original sketch's error bounds.


%% file: sections/countMin.tex
\section{Concurrent CountMin sketch}
\label{sec:countMin}

Cormode et al. propose the \emph{CountMin (CM)} sketch~\cite{cormode2005improved}, which
estimates the frequency of an item $a$, denoted $f_a$, in a data stream, where the data stream
is over some alphabet $\Sigma$. The CM sketch supports two operations: {\sc update}($a$),
which updates the object based on $a \in \Sigma$, and {\sc query}($a$), which returns
an estimate on the number of {\sc update}($a$) calls that preceded the query.

The sequential algorithm's underlying data structure is a matrix $c$ of $w \times d$ counters, for some parameters
$w,d$ determined accordingly to the desired error and probability bounds.
The sketch uses $d$ hash functions $h_i: \Sigma \mapsto [1,w]$, for $1 \leq i \leq d$.
The hash functions are generated using the random coin flip vector $\vv{c}$,
and have certain mathematical properties whose details are not essential for understanding this paper.
The algorithm's input (i.e., the schedule) is generated by a so-called \emph{weak adversary}, namely,
the input is independent of the randomly drawn hash functions.

The CountMin sketch, denoted $CM(\vv{c})$, is illustrated in Figure~\ref{img:cmSketch}, and its
pseudo-code is given in Algorithm~\ref{alg:count-min}.
On {\sc update}($a$), the sketch increments counters $c[i][h_i(a)]$ for every
$1 \leq i \leq d$. {\sc query}($a$) returns $\hat{f}_a=\min_{1 \leq i \leq d}\{c[i][h_i(a)]\}$.

\begin{algorithm}
    \begin{algorithmic}[1]

        \State array $c[1 \dots d][1 \dots w]$ \Comment{Initialized to $0$}
        \State hash functions $h_1, \dots h_d$ \Comment{$h_i: \Sigma \mapsto [1,w]$, initialized using $\vv{c}$}
        \Statex
        \Procedure{update}{$a$}
        \For{$i : 1 \leq i \leq d$}
        \State atomically increment $c[i][h_i(a)]$ \label{l:counter-inc}
        \EndFor
        \EndProcedure


        \Procedure{query}{$a$}
        \State $min \gets \infty$
        \For{$i : 1 \leq i \leq d$}
        \State $c \gets c[i][h_i(a)]$ \label{l:read-min}
        \State \algorithmicif\ $min > c$ \ \algorithmicthen\ $min \gets c$ \label{l:min-update}
        \EndFor
        \State \textbf{return} $min$
        \EndProcedure
    \end{algorithmic}
    \caption{CountMin($\vec{c}$) sketch.}
    \label{alg:count-min}
\end{algorithm}

Cormode et al. show that, for desired bounds $\delta$ and $\alpha$, given appropriate values of $w$ and $d$, with probability
at least $1-\delta$, the estimate of a query returning $\hat{f}_a$ is bounded by $f_a \leq \hat{f}_a \leq f_a + \alpha n$,
where $n$ is the the number of updates preceding the query.
Thus, for $\epsilon= \alpha n$, CM is a sequential $(\epsilon, \delta)$-bounded
object. Its sequential specification distribution is $\{CM(\vv{c})\}_{\vv{c} \in \Omega^\infty}$.

\afterpage{
  \begin{figure}
    \begin{center}
     \includegraphics[width=0.4\textwidth,trim=10 0 30 10,clip]{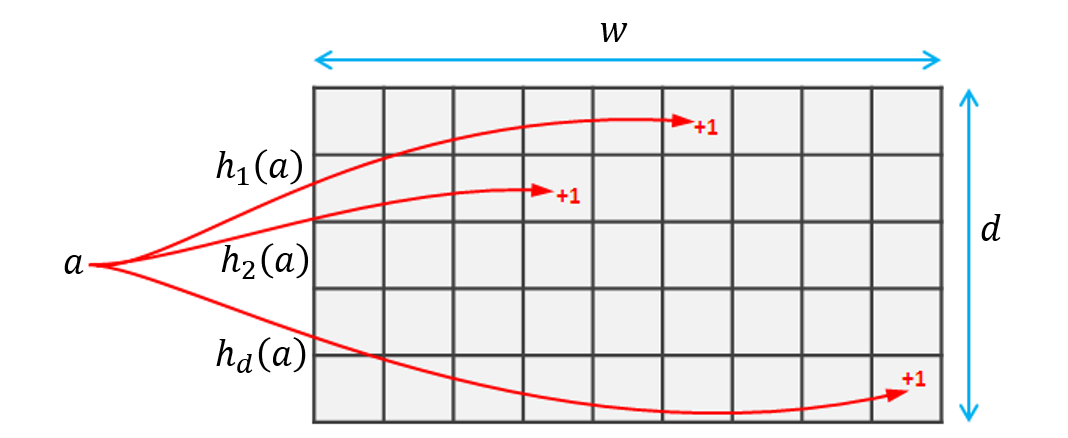}
      \caption[The LOF caption]{An example CountMin sketch, of size $w \times d$, where $h_1(a)=6$, $h_2(a)=4$ and $h_d(a)=w$.\footnotemark}
     \label{img:cmSketch}
    \end{center}
  \end{figure}
  \footnotetext{Source: \url{https://stackoverflow.com/questions/6811351/explaining-the-count-sketch-algorithm}, with alterations.}
}

Proving an error bound for an efficient parallel implementation of the CM sketch for existing criteria is not trivial.
Using the framework defined by Rinberg et al.~\cite{rinberg2019fast} requires the query to take a
strongly linearizable snapshot of the matrix~\cite{ovens2019strongly}.
Distributional linearizability~\cite{alistarh2018distributionally} necessitates an analysis of the error bounds
directly in the concurrent setting, without leveraging the sketch's existing analysis for the sequential setting.

Instead, we utilize IVL to leverage the sequential analysis for a parallelization that
is not strongly linearizable (or indeed linearizable), without using a
snapshot. Consider the straightforward parallelization of the CM sketch,
whereby the operations of Algorithm~\ref{alg:count-min} may be invoked concurrently and each counter is atomically incremented on
line~\ref{l:counter-inc} and read on line~\ref{l:read-min}. We call this parallelization $PCM(\vv{c})$. We next
prove that i is IVL.

\begin{lemma}
    $PCM$ is an IVL implementation of $CM$.
    \label{lmma:count-min-ivl}
\end{lemma}
\begin{proof}
    Let $H$ be a history of an execution $\sigma$ of $PCM$.
    Let $H_1$ be a linearization of $H^?$ such that every query is linearized prior to every
    concurrent update, and let $H_2$ be a linearization of $H^?$ such that every query is linearized after every
    concurrent update. Let $\sigma_i$ for $i=1,2$ be a sequential execution of $CM$ with history $H_i$.
    Consider some $Q=${\sc query}($a$) that returns in $H$, and let $U_1,\dots,U_k$ be the concurrent updates to $Q$.
    
    Denote by $c_\sigma(Q)[i]$ the value read by $Q$ from $c[i][h_i(a)]$ in line~\ref{l:read-min} of Algorithm~\ref{alg:count-min}
    in an execution $\sigma$.
    As processes only increment counters, for every $1 \leq i \leq d$, $c_{\sigma}(Q)[i]$ is at least
    $c_{\sigma_1}(Q)[i]$ (the value when the query starts) and at most $c_{\sigma_2}(Q)[i]$ (the value when
    all updates concurrent to the query complete). Therefore,
    $c_{\sigma_1}(Q)[i] \leq c_{\sigma}(Q)[i] \leq c_{\sigma_2}(Q)[i]$.

    Consider a randomly sampled coin flip vector $\vv{c} \in \Omega^\infty$.
    Let $j$ be the loop index the last time query $Q$ alters the value of its local variable $min$ (line~\ref{l:min-update}),
    i.e., the index of the minimum read value.
    As a query in a history of $CM(\vv{c})$ returns the minimum value in the array, $\text{ret}(Q, \tau_{CM(\vv{c})}(H_1)) \leq c_{\sigma_1}(Q)[j]$. Furthermore, $\text{ret}(Q, \tau_{CM(\vv{c})}(H_2))$
    is at least $c_{\sigma}(Q)[j]$, otherwise $Q$ would have read this value and returned it instead. Therefore:
    \[
        \text{ret}(Q, \tau_{CM(\vv{c}))}(H_1)) \leq \text{ret}(Q, H(PCM, \sigma, \vv{c})) \leq \text{ret}(Q, \tau_{CM(\vv{c})}(H_2))
    \]
    As needed.
\end{proof}

Combining Lemma~\ref{lmma:count-min-ivl} and Theorem~\ref{thm:SIVL-bound}, and by utilizing the sequential
error analysis from~\cite{cormode2005improved}, we have shown the following corollary:
\begin{corollary}
    Let $\hat{f}_a$ be a return value from query $Q$. Let $f_a^\text{start}$ be the ideal frequency of element $a$
    when the query starts, and let $f_a^\text{end}$ be the ideal frequency of element $a$ at its end, and let $\epsilon=\alpha n$ where $n$
    is the stream length at the end of the query. Then:
    \[ f_a^\text{start} \leq \hat{f}_a \leq f_a^\text{end} + \epsilon \text{ with probability at least } 1-\delta.\]
\end{corollary}

The following example demonstrates that $PCM$ is not a linearizable implementation of $CM$.

\begin{example}
    Consider the following execution $\sigma$ of $PDC$: Assume that $\vv{c}$ is such that $h_1(a)=h_2(a)=1$, $h_1(b)=2$ and $h_2(b)=1$.
    Assume that initially
    \[ c=\CM{1}{4}{2}{3}.\]
    First, process $p$ invokes $U=${\sc update}$(a)$ which increments $c[1][1]$ to $2$ and stalls.
    Then, process $p$ invokes $Q_1=${\sc query}$(a)$ which reads $c[1][1]$ and $c[2][1]$ and returns $2$,
    followed by $Q_2=${\sc query}$(b)$ which reads $c[1][2]$ and $c[2][1]$ and returns $2$. Finally, process $p$ increments $c[2][1]$ to be $3$.

    Assume by contradiction that $H$ is a linearization if $\sigma$, and $H \in CM(\vv{c})$.
    The return values imply that $U \prec_H Q_1$ and $Q_2 \prec_H U$. As $H$ is a linearization, it maintains
    the partial order of operations in $\sigma$, therefore $Q_1 \prec_H Q_2$. A contradiction.
\end{example}

%% file: sections/adderObject.tex
\section{Shared batched counter}
\label{sec:adder}

We now show an example where IVL is inherently less costly than linearizability.
In Section~\ref{ssec:ivl-adder} we present an IVL batched counter, and show that the {\sc update} operation
has step complexity $O(1)$. The algorithm uses single-writer-multi-reader(SWMR) registers.
In Section~\ref{ssec:lower-bound} we prove that all linearizable implementations
of a batched counter using SWMR registers have step complexity $\Omega(n)$ for the {\sc update} operation.

\subsection{IVL batched counter}
\label{ssec:ivl-adder}

We consider a \emph{batched counter} object, which supports the operations {\sc update}($v$) where $v \geq 0$, and {\sc read}().
The sequential specification for this object is simple: a {\sc read} operation returns the sum of all values passed to {\sc update}
operations that precede it, and $0$ if no {\sc update} operations were invoked. The {\sc update} operation returns nothing. When the
object is shared, we denote an invocation of {\sc update} by process $i$ as {\sc update}$_i$. We denote the sequential specification
of the batched counter by ${\mathcal H}$.

\begin{algorithm}
    \begin{algorithmic}[1]

        \State shared array $v[1 \dots n]$
        \Procedure{update$_i$}{$v$}
        \State $v[i] \gets v[i] + v$
        \EndProcedure


        \Procedure{read}{}
        \State $\mathit{sum} \gets 0$
        \For{$i : 1 \leq i \leq n$}
        \State $\mathit{sum} \gets \mathit{sum} + v[i]$
        \EndFor
        \State \textbf{return} $\mathit{sum}$
        \EndProcedure
    \end{algorithmic}
    \caption{Algorithm for process $p_i$, implementing an IVL batched counter.}
    \label{alg:ivl-adder}
\end{algorithm}

Algorithm~\ref{alg:ivl-adder} presents an IVL implementation for a batched counter
with $n$ processes using an array $v$ of $n$ SWMR registers.
The implementation is a trivial parallelization: an {\sc update} operation increments
the process's local
register while a {\sc read} scans all registers and returns their sum. This
implementation is not linearizable because the reader may see a later {\sc update}
and miss an earlier one, as illustrated in Figure~\ref{img:adderIVL}.
We now prove the following lemma:
\begin{lemma}
    Algorithm~\ref{alg:ivl-adder} is an IVL implementation of a batched counter.
    \label{lmma:ivl-adder}
\end{lemma}
\begin{proof}
    Let $H$ be a well-formed history of an execution $\sigma$ of Algorithm~\ref{alg:ivl-adder}.
    We first complete $H$ be adding appropriate responses to all {\sc update} operations, and removing all pending {\sc read} operations, we denote
    this completed history as $H'$.

    Let $H_1$ be a linearization of $H'^?$ given by ordering {\sc update} operations by their
    return steps, and ordering {\sc read} operations after all preceding operations in $H'^?$, and before concurrent ones. Operations
    with the same order are ordered arbitrarily.
    Let $H_2$ be a linearization of $H'^?$ given by ordering {\sc update} operations by their
    invocations, and ordering {\sc read} operations operations before all operations that precede them in $H'^?$, and after concurrent ones. Operations
    with the same order are ordered arbitrarily.
    Let $\sigma_i$ for $i=1,2$ be a sequential execution of a batched counter with history $\tau_\mathcal{H}(H_i)$.

    By construction, $H_1$ and $H_2$ are linearizations of $H'^?$. Let $R$ be some {\sc read} operation that completes
    in $H$. Let $v[1 \dots n]$ be the array as read by $R$ in $\sigma$, $v_1[1 \dots n]$ as read by $R$ in $\sigma_1$
    and $v_2[1 \dots n]$ as read by $R$ in $\sigma_2$. To show that
    $\text{ret}(R, \tau_\mathcal{H}(H_1)) \leq \text{ret}(R, H) \leq \text{ret}(R, \tau_\mathcal{H}(H_2))$,
    we show that $v_1[j] \leq v[j] \leq v_2[j]$ for every index $1 \leq j \leq n$.

    For some index $j$, only $p_j$ can increment $v[j]$. By the construction of $H_1$, all {\sc update} operations
    that precede $R$ in $H$ also precede it in $H_1$. Therefore $v_1[j] \leq v[j]$. Assume by contradiction that $v[j] > v_2[j]$.
    Consider all concurrent {\sc update} operations to $R$. After all concurrent {\sc update} operations end, the value
    of index $j$ is $v' \geq v[j] > v_2[j]$. However, by construction, $R$ is ordered after all concurrent {\sc update}
    operations in $H_2$, therefore $v' \leq v_2[j]$. This is a contradiction, and therefore $v[j] \leq v_2[j]$.

    As all entries in the array are non-negative, it follows that $\sum_{j=1}^n v_1[j] \leq \sum_{j=1}^n v[j] \leq \sum_{j=1}^n v_2[j]$, and
    therefore $\text{ret}(R, \tau_\mathcal{H}(H_1)) \leq \text{ret}(R, H) \leq \text{ret}(R, \tau_\mathcal{H}(H_2))$.
\end{proof}

Figure~\ref{img:adderIVL} shows a possible concurrent execution of Algorithm~\ref{alg:ivl-adder}.
This algorithm can efficiently implement a distributed or NUMA-friendly counter, as processes
only access their local registers thereby lowering the cost of incrementing the counter. This is of
great importance, as memory latencies are often the main bottleneck in shared object emulations~\cite{mahapatra1999processor}.
As there are no waits in
either {\sc update} or {\sc read}, it follows that the algorithm is wait-free. Furthermore, the {\sc read} step complexity
is $O(n)$, and the {\sc update} step complexity is $O(1)$. Thus, we have shown the following theorem:
\begin{theorem}
    There exists a bounded wait-free IVL implementation of a batched counter using only SWMR registers, such that the step complexity of {\sc update} is $O(1)$
    and the step complexity of {\sc read} is $O(n)$.
\end{theorem}

\begin{figure}[b]
    \centering
    \includegraphics[width=0.7\textwidth]{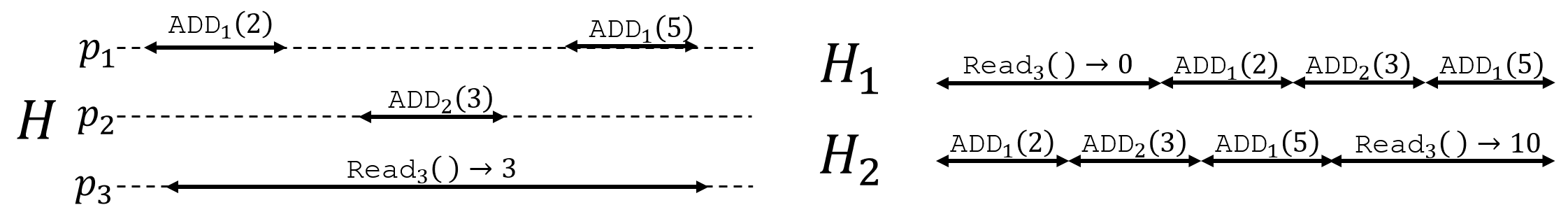}
    \caption{A possible concurrent history of the IVL batched counter: $p_1$ and
    $p_2$ update their local registers, while $p_3$ reads. $p_3$ returns an intermediate
    value between the counter's state when it starts, which is $0$, and the counter's state when it completes, which is $10$.}
    \label{img:adderIVL}
\end{figure}

%% file: sections/lowerBound.tex
\subsection{Lower bound for linearizable batched counter object}
\label{ssec:lower-bound}

The incentive for using an IVL batched counter instead of a linearizable one stems
from a lower bound on the step-complexity of a wait-free linearizable batched counter implementation from SWMR registers.
To show the lower bound we first define the binary snapshot object.
A \emph{snapshot object} has $n$ components written by separate processes, and allows a reader to
capture the shared variable states of all $n$ processes instantaneously. We
consider the \emph{binary snapshot object}, in which each state component may be either $0$ or $1$~\cite{hoepman1993binary}. The object
supports the {\sc update}$_i$($v$) and {\sc scan} operations, where the former sets the state of component $i$
to value a $v \in \{0,1\}$ and the latter returns all processes states instantaneously.
It is trivial that the {\sc scan} operation must read all states, therefore its lower bound step complexity
is $\Omega(n)$. Israeli and Shriazi~\cite{israeli1998time} show that the {\sc update} step complexity
of any implementation of a snapshot object from SWMR registers is also $\Omega(n)$. This lower bound
was shown to hold also for multi writer registers~\cite{attiya2006complexity}. While
their proof was originally given for a multi value snapshot object, it holds in the binary case as well~\cite{hoepman1993binary}.

\begin{algorithm}
    \begin{algorithmic}[1]

        \State local variable $v_i$ \Comment{Initialized to $0$}
        \State shared batched counter object $\mathit{BC}$
        \Statex
        \Procedure{update$_i$}{$v$}
        \State \algorithmicif\ $v_i = v$\ \algorithmicthen\ \textbf{return} \label{l:skip}
        \State $v_i \gets v$
        \State \algorithmicif\ $v = 1$\ \algorithmicthen\ $\mathit{BC}$.{\sc update}$_i$($2^i$) \label{l:set-1}
        \State \algorithmicif\ $v = 0$\ \algorithmicthen\ $\mathit{BC}$.{\sc update}$_i$($2^n - 2^i$) \label{l:set-0}
        \EndProcedure
        \Procedure{scan}{}
        \State $\mathit{sum} \gets \mathit{BC}$.{\sc read}() \label{l:read} 
        \State $v[0 \dots n-1] \gets [0 \dots 0]$ \Comment{Initialize an array of $0$'s}
        \For{$i : 0 \leq i \leq n-1$}
        \State \algorithmicif\ bit $i$ is set in $\mathit{sum}$\ \algorithmicthen\ $v[i] \gets 1$ \label{l:check-set}
        \EndFor
        \State \textbf{return} $v[0 \dots n-1]$
        \EndProcedure
    \end{algorithmic}
    \caption{Algorithm for process $p_i$, solving binary snapshot with a batched counter object.}
    \label{alg:bs-with-adder}
\end{algorithm}

To show a lower bound on the {\sc update} operation of wait-free linearizable batched counters,
we show a reduction from a binary snapshot to a batched counter in
Algorithm~\ref{alg:bs-with-adder}. It uses a local variable $v_i$ and a shared batched counter object.
In a nutshell, the idea is to encode the value of the $i^\text{th}$ component
of the binary snapshot using the $i^\text{th}$ least significant bit of the counter.
When the component changes from $0$ to $1$, {\sc update}$_i$ adds $2^i$, and when it changes from $1$ to $0$,
{\sc update}$_i$ adds $2^n - 2^i$. We now prove the following invariant:
\begin{invariant}
    At any point $t$ in history $H$ of a sequential execution of Algorithm~\ref{alg:bs-with-adder},
    the sum held by the counter is $c \cdot 2^n + \sum_{i=0}^{n-1}v_i2^i$,
    such that $v_i$ is the parameter passed to the last invocation of {\sc update}$_i$ in $H'$ before $t$ if such invocation
    exists, and $0$ otherwise, for some integer $c \in \mathbb{N}$.
    \label{inv:sum}
\end{invariant}
\begin{proof}
    We prove the invariant by induction on the length of $H$, i.e., the number of invocations in $H$,
    denoted $t$. As $H$ is a sequential history, each invocation is followed by a response.
    \par{\textbf{Base:}} The base if for $t=0$, i.e., $H$ is the empty execution. In this case no updates
    have been invoked, therefore $v_i=0$ for all $0 \leq i \leq n-1$. The sum returned by the counter
    is $0$. Choosing $c=0$ satisfies the invariant.
    \par{\textbf{Induction step:}} Our induction hypothesis is that the invariant holds for a history of length $t$.
    We prove that it holds for a history of length $t+1$. The last invocation can be either a {\sc scan}, or an {\sc update}($v$)
    by some process $p_i$. If it is a {\sc scan}, then the counter value doesn't change and the invariant
    holds. Otherwise, it is an {\sc update}($v$). Here, we note two cases. Let $v_i$ be $p_i$'s value
    prior to the {\sc update}($v$) invocation. If $v = v_i$, then the {\sc update} returns without altering the sum
    and the invariant holds. Otherwise, $v \neq v_i$. We analyze two cases, $v=1$ and $v=0$. If $v=1$, then $v_i=0$.
    The sum after the update is $c \cdot 2^n + \sum_{i=0}^{n-1}v_i2^i + 2^i=c \cdot 2^n + \sum_{i=0}^{n-1}v_i'2^i$, where
    $v_j'=v_j$ if $j \neq i$, and $v'_i = 1$, and the invariant holds. If $v=0$, then $v_i=1$.
    The sum after the update is $c \cdot 2^n + \sum_{i=0}^{n-1}v_i2^i + 2^n - 2^i = (c+1) \cdot 2^n + \sum_{i=0}^{n-1}v_i'2^i$,
    where $v_j'=v_j$ if $j \neq i$, and $v'_i = 1$, and the invariant holds.
\end{proof}

Using the invariant, we prove the following lemma:
\begin{lemma}
    For any sequential history $H$, if a {\sc scan} returns $v_i$, and {\sc update}$_i$($v$) is the last update invocation in $H$
    prior to the {\sc scan}, then $v_i = v$. If no such update exists, then $v_i=0$.
    \label{lmma:scan-correctness}
\end{lemma}
\begin{proof}
    Let $S$ be a {\sc scan} in $H'$. Consider the sum $\textit{sum}$ as read by scan $S$.
    From Invariant~\ref{inv:sum}, the value held by the counter is $c \cdot 2^n + \sum_{i=0}^{n-1}v_i2^i$.
    There are two cases, either there is an update invocation prior to $S$, or there isn't. If there isn't, then by
    Invariant~\ref{inv:sum} the corresponding $v_i=0$. The process sees bit $i=0$,
    and will return $0$. Therefore, the lemma holds.

    Otherwise, there is a an update prior to $S$ in $H$. As the sum is equal to $c \cdot 2^n + \sum_{i=0}^{n-1}v_i2^i$,
    by Invariant~\ref{inv:sum}, bit $i$ is equal to $1$ iff the parameter passed to the last invocation of update was $1$.
    Therefore, the scan returns the parameter of the last update and the lemma holds.
\end{proof}

\begin{lemma}
    Algorithm~\ref{alg:bs-with-adder} implements a linearizable binary snapshot using a linearizable batched counter.
    \label{lmma:reduction}
\end{lemma}
\begin{proof}
    Let $H$ be a history of Algorithm~\ref{alg:bs-with-adder}, and let $H'$
    be $H$ where each operation is linearized at its access to the linearizable batched counter, or
    its response if $v_i = v$ on line~\ref{l:skip}.
    Applying Lemma~\ref{lmma:scan-correctness} to $H'$, we get $H' \in \mathcal{H}$ and therefore $H$ is linearizable.
\end{proof}

It follows from the algorithm that if the counter
object is bounded wait-free then the {\sc scan} and {\sc update} operations are bounded wait-free. Therefore, the lower
bound proved by Israeli and Shriazi~\cite{israeli1998time} holds, and the {\sc update} must take $\Omega(n)$
steps. Other than the access to the counter in the {\sc update} operation, it takes
$O(1)$ steps. Therefore, the access to the counter object must take $\Omega(n)$ steps. We have proven the following theorem.
\begin{theorem}
    For any linearizable wait-free implementation of a batched counter object with $n$ processes from SWMR registers, the step-complexity
    of the {\sc update} operation is $\Omega(n)$.
    \label{thm:lower-bound}
\end{theorem}

%% file: sections/conclusion.tex
\section{Conclusion}
\label{sec:conclusion}

We have presented IVL, a new correctness criterion that provides flexibility
in the return values of quantitative objects while bounding the error that this may
introduce. IVL has a number of desirable properties: First,
like linearizability, it is a local property, allowing designers to reason about each part
of the system separately. Second, also like linearizability but unlike other relaxations of
it, IVL preserves the error bounds of PAC objects. Third, IVL is generically defined for
all quantitative objects, and does not necessitate object-specific definitions. Finally,
IVL is inherently amenable to cheaper implementations than linearizability in some cases.

Via the example of a CountMin sketch, we have illustrated that IVL provides
a generic way to efficiently parallelize data sketches while leveraging their
sequential error analysis to bound the error in the concurrent implementation.

%% file: sections/appendix.tex
\pagestyle{empty}

\section*{Appendix}

\section{Locality proof}
\label{sec:locality-proof}
We now prove that IVL is a local property:
\local
\begin{proof}
    Let $A$ be a deterministic algorithm, and let $\mathcal{H}_x$ be the sequential specification of object $x$, for every $x \in \mathcal{X}$.
    The ``only if'' part is immediate.
  
    Denote by ${H_1^x}, {H_2^x}$ the linearizations of ${H|_x}^?$ guaranteed by the
    definition of IVL.
    We first construct a linearization $H_1$ of $H^?$, defined
    by the order $\prec_{H_1}$ as follows: For every pending operation on object $x$, we either
    add the corresponding response or remove it based on ${H_1^x}$. We then construct a partial order
    of operations as the union of $\{{\prec}_{H_1^x}\}_{x \in \mathcal{X}}$ and the realtime order
    of operations in $H$. As ${\prec}_{H_1^x}$ must adhere to the realtime order of $H|_x$,
    and therefore $H$, and the order of operations in $\{{\prec}_{H_1^x}\}_{x \in \mathcal{X}}$ are disjoint, this partial order
    is well defined. Consider two concurrent operations $op_1, op_2$ in $H$. If they
    do not belong to the same history $H|_x$ for some object $x$, we order them arbitrarily in $\prec_{H_1}$.
    We construct linearization $H_2$ of $H^2$ by defining the order of operations $\prec_2$ in a similar fashion.
  
    By construction all invocations and responses appearing in $H^?$ appear both in $H_1$ and in $H_2$,
    and $H_1$ and $H_2$ preserve the partial order $\prec_{H^?}$. Therefore, $H_1$ and $H_2$ are linearizations
    of $H^?$.
    
    Consider some read $R$ on some object $x \in \mathcal{X}$ that returns in $H$. As $H|_x$ is IVL
    $\text{ret}(R,  \tau_{\mathcal{H}_x}(H_1^x) \leq \text{ret}(R, H|_x) \leq \text{ret}(R, \tau_{\mathcal{H}_x}(H_2^x))$.
    Note that $\text{ret}(R, H|_x) = \text{ret}(R, H)$. Furthermore, $\text{ret}(R, \tau_{\mathcal{H}_x}(H_i^x))= \text{ret}(R, \tau_{\mathcal{H}}(H_i))$ for $i \in \{1,2\}$,
    as objects other than $x$ do not affect the return value of this operation. Therefore $H$ is IVL.
  \end{proof}